\documentclass[submission,copyright,creativecommons]{eptcs}
\usepackage{amssymb,amsmath, amsthm, amsfonts}
\usepackage{graphicx,mathrsfs,bbold,color}
\usepackage{tikz}
\usepackage{tikz-cd}
\usepackage{float}
\newtheorem*{theorem*}{Theorem}
\newtheorem{theorem}{Theorem}
\newtheorem{proposition}{Proposition}

\newtheorem{fact}{Fact}

\newtheorem{example}{Example}

\newtheorem*{lemma*}{Lemma}
\newtheorem*{remark*}{Remark}
\usepackage{csquotes}

\usepackage{color}
\newcommand{\CA}{\mathcal{C}({\bf A})}
\newcommand{\atom}{{\rm at}}
\newcommand{\atomA}{\atom({\bf A})}

\newcommand{\amrand}[2]
{\leavevmode
	\marginpar
	[\raggedleft\scriptsize \leavevmode\normalcolor #1: #2]
	{\raggedright\scriptsize \leavevmode\normalcolor #1: #2}
}

\newcommand{\cond}{\vartriangleright}
\usepackage{iftex}

\ifpdf
  \usepackage{underscore}         
  \usepackage[T1]{fontenc}        
\else
  \usepackage{breakurl}           
\fi

\title{Conditionals Based on Selection Functions, Modal Operators and Probabilities}
\author{Tommaso Flaminio 
\institute{IIIA - CSIC\\ 
Campus de la UAB\\
Barcelona, Spain}
\email{tommaso@iiia.csic.es}
\and
Lluis Godo
\institute{IIIA - CSIC\\ 
Campus de la UAB\\
Barcelona, Spain}
\email{godo@iiia.csic.es}
\and
Giuliano Rosella 
\institute{Department of Philosophy\\
University of Turin\\
Turin, Italy}
\email{giuliano.rosella@unito.it}
}
\def\titlerunning{Conditionals Based on Selection Functions}
\def\authorrunning{T. Flaminio, L. Godo \& G. Rosella}
\begin{document}
\maketitle

\section{Introduction}

Methods for probability updating, of which Bayesian conditionalization is the most well-known and widely used, are modeling tools that aim to represent the process of modifying an initial epistemic state, typically represented by a prior probability function $P$, which is adjusted in light of new information. Notably, updating methods and conditional sentences seem to intuitively share a deep connection, as is evident in the case of conditionalization. Starting with a probability $P$, the question of what the probability of an event $B$ is, given that $A$ holds (i.e., $P(B \mid A)$), appears to involve conditional reasoning of the form "if $A$ then $B$". Indeed, Adams \cite{Adams1975} argues that the assertability of (indicative) conditionals, represented as "if $A$ then $B$", aligns with the corresponding conditional probability $P(B \mid A)$. However, this intuitive connection between conditionals and updating methods has been challenged by Lewis's celebrated triviality result \cite{Lewis1976}. This result demonstrates that conditional probability $P(B \mid A)$ cannot generally be equated with the probability of the corresponding conditional, $P(A \cond B)$, without leading to trivializing constraints on the initial probability functions.

Building upon these foundational works, research on the relationship between updating methods and conditionals has been and continues to be prolific. This research has primarily focused on understanding whether and how updating methods can be semantically represented in terms of conditional connectives, or conversely, what kind of updating methods are associated with specific semantic conditional operators (see, for instance, \cite{Rosella2023-ROSCAM-10, FRS2025, Hajek1989, Fraassen1976, EgreForthcoming-EGRCAU, Santorioforth}).

The present work contributes to this line of research and aims at shedding new light on the relationship between updating methods and conditional connectives. Departing from previous literature that often focused on a specific type of conditional or a particular updating method, our goal is to prove general results concerning the connection between conditionals and their probabilities. This will allow us to characterize the probabilities of certain conditional connectives and to understand what class of updating procedures can be represented using specific conditional connectives. Broadly, we adopt a general perspective that encompasses a large class of conditionals and a wide range of updating methods, enabling us to prove some general results concerning their interrelation.

Let us begin by providing some background. As observed in \cite{Guenther2022PCIP} (see also \cite{FRS2025}), conditionalization updating can be situated within the broader context of {\em imaging} updating methods. Thge imaging method was initially introduced by Lewis \cite{Lewis1976} as a way to circumvent his triviality result and it hs been extended and generalized by several authors, notably by G\"ardenfors \cite{Gardenfors1982}. While conditionalization is an updating method that requires only an algebra of events and a (positive) probability function, imaging methods update an a prior probability by employing a {\em closeness relation} among possible worlds (the atoms of the algebra), often interpreted as a measure of similarity. This closeness relation imposes a certain structure on the underlying set of possible worlds that is employed within the context of imaging updating methods to determine which possible worlds the probability mass should be redistributed to, typically the closest (most similar) ones to a given one.

Such a closeness relation can also be used to specify the truth conditions of a wide range of conditional connectives, such as variably strict conditionals \cite{Lewis1973-CF} and preferential conditionals \cite{Burgess1981-JOHQCP, Negri2015, Girlando2021}. The informal idea is that a conditional $a \cond b$ is true at a world $w$ if and only if $b$ is true in the closest $a$-worlds to $w$, i.e., those closest worlds that make $a$ true. The structure of closeness has been represented through various model-theoretic frameworks, among which ordered models and sphere models are particularly noteworthy. Sphere models were prominently used by Lewis \cite{Lewis1973-CF} as the intended semantic device for specifying the truth conditions of variably strict conditionals (including counterfactual conditionals), while ordered models have been used, for example, to specify the truth conditions of preferential conditionals. Consequently, this closeness relation of possible worlds has been employed both in specifying imaging-like updating methods and in establishing a semantics for a large class of conditional connectives. Hence, intuitively, one might expect a deep relationship between imaging-like updating methods and conditionals based on a closeness relation, given that both rely on the same  structure. 

Our investigation will specifically focus on these two dimensions: imaging-like updating methods and conditionals, both having a semantics grounded on a closeness relation. Given our aim for generality, we will slightly deviate from the standard model-theoretic approach to conditionals and updating methods that relies on sphere models or ordered models to represent closeness relations, and we will utilize models based on {\em selection functions}.  
Notably, certain classes of ordered models and sphere models have equivalent representations in terms of selection functions (see, e.g., \cite{Lewis1971} and also \cite{Grahne98}). One of the main advantages of the selection function is its technical utility: it is a very general tool for representing a general structure over worlds. However, one of its main shortcomings is conceptual: while orders and spheres offer an intuitive representation of a certain closeness relation, the relation induced by a selection function can sometimes be obscure and abstract.

In the present contribution we aim to provide a further generalization of imaging methods by relying on this more general setting of selection functions. Simultaneously, we aim to present results that conceptually enhance the understanding of certain updating procedures and selection function models. Inspired by this generality, we will primarily work within an algebraic setting, which also generalizes possible worlds models and imaging updating methods. Our basic framework will consist of a finite Boolean algebra $\mathbf{A}$, whose domain is denoted by the corresponding non-bold italic letter $A$,  and whose atoms, forming the set $\atom({\bf A})$, will be thought as possible worlds. The finite algebra ${\bf A}$ will be equipped with a probability function $P$, whence $\atom({\bf A})$ will be endowed with a probability distribution that we still denote by $P$. The structure on $A$ that we need to establish our general updating methods is represented by two basic tools (more technical details will be provided in the next sections):

\begin{itemize}
\item a selection function $f: A \times \atom(\mathbf{A}) \to A$ that, for every atom $\alpha \in \atom(\mathbf{A})$ and every $a \in A$, specifies an element of $A$ that correspond to the set of worlds closest to $\alpha$ according to $a$;

\item a function $\lambda: A \times \atom(\mathbf{A}) \to [0,1]^{|\atom(\mathbf{A})|}$ that will be used to determine how to distribute the probability mass of each $\alpha \in \atom(\mathbf{A})$ to its closest worlds $f(a, \alpha)$.
\end{itemize}
While the next section will introduce the basic notions and results that constitute the ground of our investigation, Section \ref{sec3} will connect selection function-based conditionals with modalities and probabilities, while Section \ref{sec4} will more directly deal with general updating methods. We will end the present paper with some conclusions that will be put forward in Section \ref{sec5}.

\section{Selection Functions and Conditionals}

Let us begin by establishing the fundamental components that define our framework. Boolean algebras will be presented in their usual language $(\wedge,\vee,\neg,\bot,\top)$ of type $(2,2,1,0,0)$ and where the implication operator $\to$ will be defined as $a\to b=\neg a\vee b$. All Boolean algebras that we consider in this paper are assumed to be finite; a generic Boolean algebra will be denoted by ${\bf A}$ and the set of its atoms is given by $\atomA = \{\alpha_1,\ldots, \alpha_n\}$. Let $\wp(X)$ stand for the powerset of a set $X$. Thus, ${\bf A}$ and $\wp(\atomA)$ are isomorphic algebras whence we shall henceforth identify  elements of $A$ with subsets of $\atomA$ and use a typical set-theoretical notation without danger of confusion. For instance, for $a\in A$, we will write $|a|$ to denote the cardinality of the subset of $\atomA$ below $a$ itself. Moreover, for $a,b\in A$ we write $a\subseteq b$ to state that $a$ lies below $b$ in the lattice order of ${\bf A}$ while, when $\alpha\in \atomA$ and $a\in A$ we will simply write $\alpha\in a$ to denote that $\alpha$ is an atom among those that lie below $a$.

A selection function for $\mathbf{A}$ is a mapping $f: A\times \atom({\bf A})\to A$ that associate to each pair of an atom and an element of $A$ another element of $A$ that we will identified with a set of atoms of $\mathbf{A}$ as observed above. Initially, we do not impose any further properties on this function. Intuitively, the selection function specifies the set of closest worlds (atoms) to a given world, relative to a certain element. For instance, $f(a, \alpha)$ represents the set of worlds closest to $\alpha$, as determined by $a$. For every selection function $f: A\times \atom({\bf A})\to A$, let us define an {\em selection function-based conditional} operator (or simply a {\em conditional} operator)
$$
\cond_f: A\times A\to A
$$ 
as follows: for all $a,b\in A$, 
\begin{equation}\label{eqCond}
a\cond_f b=\{\alpha\in \atom({\bf A})\mid  f(a, \alpha) \subseteq  b\}.
\end{equation}
In other words, for all $\alpha\in \atom({\bf A})$,
\begin{center}
$\alpha \in   a\cond_f b\mbox{ iff } f(a, \alpha)\subseteq  b$.
\end{center}

\noindent
As we have already anticipated, this algebraic setting has a very intuitive possible-worlds semantic counterpart: an atom $\alpha$  belonging to  $a \cond_f b$ can be regarded as the possible world $\alpha$ making $a \cond_f b$ true. Consequently, given an intuitive interpretation of the selection function in terms of a closeness relation, the above condition is telling us that $a \cond_f b$ is true at $\alpha$ if and only if all the closest $a$-worlds make $b$ true.  

This closeness relation, specified by $f$, is at this stage very general while it can be instantiated in different ways to represent, for instance, an order of similarity \cite{Lewis1973-CF} or normality among worlds \cite{Burgess1981-JOHQCP}. The logical properties of $\cond_f$ will be induced by the constraints imposed on $f$, i.e., the type of closeness relation we aim to represent. Indeed, having imposed no specific properties on the selection function $f$ does not inherently explain why the operator $\cond_f$ defined as above should be considered a \emph{conditional} operator. For example, it might be desirable that $a \cond_f a$ is always true, or that $\cond_f$ satisfies the modus ponens inference rule, and so on. All these properties can be enforced by imposing relevant constraints on the closeness relation involved in evaluating the conditional, namely the selection function. Below, we list several relevant constraints, including some that are discussed in \cite{Lewis1971}: consider a finite Boolean algebra $\mathbf{A}$ and a selection function $f: A\times \atomA\to A$, we can define the following properties for all $\alpha \in \atomA$:

    \begin{itemize}

    \item $f(\bot, \alpha)=\bot$ \hfill (emptyness)
    
    \item for $a \in A\setminus\{\bot\}$, $f(a, \alpha) \neq \bot$   \hfill (normality)

    \item for $a \in A$, $  f(a, \alpha) \subseteq  a$ \hfill (identity)

    \item if $\alpha \in   a$, then $\alpha \in f(a, \alpha)$ \hfill (centering-1)

    \item if $\alpha \in   a$, then $ f(a, \alpha) \subseteq  \alpha$ \hfill (centering-2)

    \item if $\alpha \in   a$, $f(a, \alpha)=\alpha$ \hfill (centering)

    \item $|f(a, \alpha)| \leq 1$ \hfill (uniqueness)

    \item if $ f(a, \alpha) \subseteq  b$ and $ f(b, \alpha) \subseteq  a$, then $f(a, \alpha)=f(b, \alpha)$ \hfill (well-order)

    \item $ f(a \vee b, \alpha) \subseteq  a$ or  $ f(a \vee b, \alpha) \subseteq  b$ or $f(a \vee b, \alpha)=f(a, \alpha) \vee f(b, \alpha)$ \hfill (nesting)

    \end{itemize}

    The connection to the corresponding conditional connective is straightforward, and this allows  to prove the following result, which readily follows from the conditions imposed on the selection function and equation (1) and from well-known results proved in \cite{Lewis1971}:

    \begin{fact}\label{fc:propoff}
    Consider a finite Boolean algebra $\mathbf{A}$ and a selection function $f: A\times \atom({\bf A})\to A$; the following hold for all $a, b, c \in A$, and all $d \in A\setminus\{\bot\}$:
    \[\begin{array}{lcl}
    \bot \cond_f a = \top & \Leftrightarrow & \text{$f$ satisfies emptyness}\\
    d \cond_f b \subseteq  \neg(d \cond_f  \neg b) & \Leftrightarrow & \text{$f$ satisfies normality}\\
    a \cond_f  a = \top & \Leftrightarrow & \text{$f$ satisfies identity}\\
    a \cond_f  b \subseteq  \neg a \vee b & \Leftrightarrow & \text{$f$ satisfies centering-1}\\
    a \wedge b \subseteq  a \cond_f  b & \Leftrightarrow & \text{$f$ satisfies centering-2}\\
    a \wedge b \subseteq  a \cond_f  b \subseteq  \neg a \vee b & \Leftrightarrow & \text{$f$ satisfies centering}\\
    (a \cond_f \neg b) \vee (a \cond_f b)= \top & \Leftrightarrow & \text{$f$ satisfies uniqueness} \\
    (a \cond_f b) \wedge (b \cond_f a) \subseteq  (a \cond_f c) \leftrightarrow (b \cond_f c) & \Leftrightarrow & \text{$f$ satisfies well-order}\\
    ((a \vee b)\cond_f a) \vee ((a \vee b)\cond_f b) \vee ( ((a \vee b)\cond_f c) = ((a \cond_f c) \wedge (b \cond_f c))) & \Leftrightarrow & \text{$f$ satisfies nesting}
    \end{array}\]

    \end{fact}

    For convenience, we may refer to $\cond_f$ as to a {\em counterfactual conditional} when $f$ satisfies the following properties: (i) identity, (ii) well-ordering, (iii) nesting, and (iv) centering. We call $\cond_f$ a {\em variably strict conditional} when $f$ satisfies (i) identity, (ii) well-ordering, and (iii) nesting. Finally, we use the term \emph{Stalnaker conditional} for $\cond_f$ when $f$ satisfies (i) identity, (ii) well-ordering, (iii) nesting, (iv) centering, and (v) uniqueness. This terminology is justified by well-established results in the literature implying that that when $\cond_f$ is a counterfactual conditional, it satisfies precisely the logical principles of Lewis's counterfactuals, namely it obeys the conditional logic C1 (see \cite{Lewis1971}). When $\cond_f$ is a variably strict conditional, it obeys the conditional logic C0 (see \cite{Lewis1971}), which is the basic variably strict conditional logic (see \cite{Lewis1973-CF}). And when $\cond_f$ is a Stalnaker conditional, it obeys the logic C2 (see \cite{Lewis1971}), which corresponds to Stalnaker's logic of conditionals \cite{Stalnaker1968}.

\section{Modalities and Probabilities}\label{sec3}
In this section, we will establish some key results concerning the relationship between normal modal operators and selection functions. These results will subsequently be employed as a tool to begin exploring the connection between selection function-based conditionals and probability.
 
Consider a finite Boolean algebra $\mathbf{A}$ and a selection function $f: A\times \atom({\bf A})\to A$; let  the binary relation $R^f_a\subseteq \atom({\bf A})\times\atom({\bf A})$ be defined as follows: for all $\alpha\in A$, $R^f_a(\alpha)=f(a, \alpha)$. In other words, 
\begin{center}
$\alpha R^f_a \beta$ iff $\beta\in f(a, \alpha)$. 
\end{center}
For each such $R_a^f$, let $\Box_a^f: A \to A$ be defined as standard normal modal operator: for all $b\in A$
$$
\Box_a^f (b)= \{\alpha\in \atom({\bf A})\mid  R_a^f(\alpha)\subseteq  b\} =\{\alpha\in \atom({\bf A})\mid  f(a,\alpha)\subseteq  b\}.
$$ 
Notice that $\Box_a^f$ is, indeed, a normal modal operator in the sense of \cite[Definition 1.42]{Blackburn2002}, that is to say, $\Box_a^f\top=\top$ and $\Box_a^f(b\to c)\subseteq (\Box_a^f b\to \Box_a^f c)$. Moreover, since $R_a^f(\alpha)=f(a, \alpha)$, comparing the latter with (\ref{eqCond}) we can immediately show the following result:
\begin{fact}\label{fact0}
For a finite Boolean algebra ${\bf A}$ and a selection function $f: A\times \atom({\bf A})\to A$; for all $a, b\in A$, it holds that
$$
\alpha \in   a\cond_f b \mbox{ iff }\alpha \in   \Box^f_a(b).
$$
\end{fact}
In other words, for every $a\in A$, $a\cond_f(\cdot)$ and $\Box^f_a(\cdot)$ are the same unary operator on $A$. This implies that selection function-based conditionals can also be equivalently defined in terms of a family of indexed normal modal operators, $\{\Box_a^f \mid a \in A\}$, each of which behaves as a conditional with a fixed antecedent. Consequently, the structure $({\bf A}, \{\Box^f_a\}_{a \in A})$ forms a {\em Boolean algebra with operators} in the sense of \cite{Blackburn2002}. 
Following standard practice in modal logic, we define $\Diamond_a^f x$ as $\neg\Box_a^f \neg x$ and for all $b\in A$
$$
\Diamond_a^f(b)= \{\beta\in \atom({\bf A})\mid   R_a^f(\beta)\wedge b\neq\bot\}.
$$ 
Treating a conditional (binary) operator as a family of indexed normal modal (unary) operators provides a valuable framework for gaining new insights into the properties of conditional operators. This approach allows us to leverage the well-established behavior of normal modal operators and to demonstrate relevant results that connect the properties of a selection function $f$ (and consequently, the corresponding conditional $\cond_f$) with the properties of the induced indexed modal operators. The following result contributes to this direction:
\begin{fact}\label{fact1} 
For a finite Boolean algebra $\mathbf{A}$ and selection function $f: A\times \atom({\bf A})\to A$ the following holds for all $a \in A$: 
$$
|f(a,\alpha)|=1 \Leftrightarrow (\text{for all }x \in A, \ \Box^f_ax = \Diamond^f_ax).
$$
\end{fact} 
\begin{proof}
The identity $\Box x= \Diamond x$ holds in a Boolean algebra with a modal operator $({\bf A}, \Box)$ iff its associated Kriple frame $(\atom({\bf A}), R)$ satisfies that for all $\alpha\in \atom({\bf A})$, there exists a unique $\beta\in \atom({\bf A})$ such that $\alpha R \beta$, and hence $R$ is a function from $\atom({\bf A})$ to itself, \cite{Blackburn2002}. Thus in particular,  $\Box^f_ax = \Diamond^f_ax$ for all $x \in A$ if and only if $R^f_a$ is a function and hence, if and only if, for all $\alpha$, there is a unique $\beta$ such that $\alpha R_a^f\beta$, that is to say, there exists a unique $\beta\in f(a,\alpha)$. Therefore, if and only if $|f(a,\alpha)|=1$.
\end{proof}

We now have the basic ingredients to begin investigating the relationship between conditionals and probability using the introduced general framework. For our investigation we will henceforth always assume  that  probability functions are positive, {\em i.e.}, they assign $0$ only to the impossible event. Given the finite algebraic environment we are adopting, this choice comes at no cost. So, let $P:{\bf A}\to[0,1]$ be a positive probability on a finite Boolean algebra ${\bf A}$ and consider a selection function $f: A\times \atom({\bf A})\to A$. Then, given the above assumptions and definitions, we define
$$
P(a\cond_f b)=\sum_{\alpha\in  a\cond_f b}P(\alpha).
$$
That is, the probability of a conditional $a \cond_f b$ is given by the sum of the probabilities of the atomic states (worlds) where the conditional is satisfied. In other words, $P(a \cond_f b)$ represents the probability that the conditional $a \cond_f b$ is true.
Directly from Fact \ref{fact0}, we obtain the following.
\begin{fact}\label{fact2}
For a finite Boolean algebra ${\bf A}$, a selection function $f: A\times \atom({\bf A})\to A$, and probability $P:{\bf A}\to[0,1]$ the following holds for all $a, b\in A$:
$$
P(a\cond_f b)=P(\Box_a^f (b)).
$$
\end{fact}
\noindent
As a direct consequence of the above Fact \ref{fact2} and by the main results of \cite{Harmanec1994}, the map defined as
\begin{equation}\label{eqBELIEF}
P(a\cond_f(\cdot))=P(\Box_a^f(\cdot))
\end{equation}
is a {\em belief function} in the sense of Dempster-Shafer theory \cite{Dempster}. Clearly, this observation does not preclude the possibility that $P(a\cond_f(\cdot))$ is a probability function for some $a \in A$.

Let us now examine in more detail this connection between conditionals and belief functions. First, recall that belief functions differ from probability functions in their treatment of disjunction: probabilities are additive, satisfying $P(a \vee b)=P(a) + P(b) - P(a \wedge b)$, whereas belief functions, denoted $Bel$, are superadditive, that is to say, $Bel(a\vee b)\geq Bel(a)+Bel(b)-Bel(a\wedge b)$. 

 Besides the axiomatic definition of belief functions, that the interest reader can find for instance in \cite{Halpern} and on which we will not rely for the rest of the present paper, it is convenient to recall how belief functions can be described in terms of their {\em mass distribution}. 
Specifically, every belief function is induced by a mass distribution $m$ over a Boolean algebra $\mathbf{A}$, i.e., a function $m: A\to [0, 1]$ such that $\sum_{a \in A} m(a)=1$. From this mass distribution, the belief function is then defined as follows:
\begin{equation}\label{def:belief}
Bel(a)=\sum_{b \subseteq a} m(a).
\end{equation}
 Observe that, while probability functions on ${\bf A}$ are fully determined by distributions on $\atomA$, belief functions on $\bf A$ are fully described by mass functions on $A$ itself, and not only on its atoms. This easy observation marks the key difference between the two uncertain measures.

Given this definition, we can more vividly represent the belief function induced by a conditional $\cond_f$ according to equation \ref{eqBELIEF} by appealing to a generalization of the \emph{imaging-like updating method} for belief functions introduced in \cite{Dubois1994}. Specifically, let us consider a finite Boolean algebra $\mathbf{A}$, a selection function $f: A \times \atom(\mathbf{A})\to A$, and a probability $P: A \to [0, 1]$. For a given element $a \in A$ we can define a mass distribution $m_a: A \to [0, 1]$  as follows: $m_a(b)=\sum_{ \alpha:f(a, \alpha) = b}P(\alpha)$. Then the corresponding \emph{imaged} belief function is:
\begin{equation}\label{def:imagedmass}
Bel_a(b)=\sum_{c \subseteq b}m_a(c)=\sum_{c \subseteq b}\left(\sum\limits_{\alpha:  f(a, \alpha) = c}P(\alpha)\right)=\sum\limits_{\alpha:  f(a, \alpha) \subseteq b}P(\alpha).
\end{equation}

\noindent
As a consequence, we obtain the following:

\begin{fact}\label{fc:characterization}
Consider a finite Boolean algebra $\mathbf{A}$, a selection function $f: A \times \atom(\mathbf{A})\to A$, and a probability $P: {\bf A} \to [0, 1]$, the following holds for all $a, b \in A$:
\[P(a\cond_f (\cdot))=P(\Box_a^f (\cdot))=Bel_a(\cdot)\]
\[P(a\cond_f b)=P(\Box_a^f (b))=Bel_a(b)\]

\end{fact}

Specifically, the probability of a selection function conditional $a \cond_f b$ is equal to the corresponding belief function of the consequent imaged on the antecedent, i.e. $Bel_a(b)$.

This general result implies that the probability of a broad class of conditionals can be characterized in terms of an updated belief function. For instance, variably strict conditionals, which encompass both Stalnaker conditionals and Lewis counterfactuals (see \cite{Lewis1973-CF, Lewis1971}), are known to be characterizable as selection function-based conditionals. Similarly, preferential conditionals \cite{Negri2015, Girlando2021, Burgess1981-JOHQCP} also fall under this category. Indeed, the semantics of preferential conditionals is typically defined using preorders, such that a preferential conditional $a \cond_f b$ is true at a world $\alpha$ if and only if the minimal $a$-worlds with respect to the preorder associated with $\alpha$ are also $b$-worlds. This semantics can be reformulated such that the set of minimal $a$-worlds, with respect to the preorder associated with $\alpha$, constitutes a selection function $f(a,\alpha)$. Hence, the aforementioned result also holds for preferential conditionals.  This result generalizes Lewis's \cite{Lewis1976} work on Stalnaker conditionals. Lewis demonstrated that the probability of a Stalnaker conditional $a \cond_f b$ corresponds to the probability of the consequent $b$ imaged on the antecedent $a$. Essentially, the probability of a Stalnaker conditional aligns with a specific updated probability distribution. Given that belief functions generalize probabilities, our finding indicates that the probability of selection function-based conditionals generally follows an updated belief function, and in specific instances, for instance in the case of Stalnaker conditionals, this updated belief function can indeed reduce to a probability, thereby implying Lewis's original result (further details are in subsequent sections).

Therefore, it is legitimate to ask under what conditions the belief function induced by a conditional is also a probability function. The next result provides an answer to this question, bringing together some of the results seen so far and characterizing the conditions under which the map defined in (\ref{eqBELIEF}) is a probability function.
\begin{proposition}\label{propBasic1}
Consider a finite Boolean algebra ${\bf A}$, a selection function $f: A \times \atom(\mathbf{A})\to A$, and a  positive probability $P: A \to [0, 1]$, for all $a\in A$, the following conditions are equivalent:
\begin{itemize}
\item[(i)] $P(a\cond_f(\cdot)):A\to[0,1]$ is a probability function,
\item[(ii)] $|f(a, \alpha)|=1$, 
\item[(iii)] $R_a^f$ is a function,
\item[(iv)] for all $x \in A$, $\Box^f_ax = \Diamond^f_ax$.
\end{itemize}
\end{proposition}
\begin{proof}
The equivalence between $(ii)$, $(iii)$, and $(iv)$ follows from Fact \ref{fact1} and the definition of $R_a^f$. We now prove that $(iv)$ implies $(i)$. Indeed, if $\Box^f_ax = \Diamond^f_ax$ holds for all $x \in A$, we have that $\Box_a^f$ distributes over disjunction, namely $\Box_a^f(c \vee b)=\Box_a^fc \vee \Box^f_a b$, and so, by Fact \ref{fc:characterization}, we would have that $Bel_a(c \vee b)=Bel_a(b) + Bel(b) - Bel_a(c \wedge b)$. Namely, the belief function $Bel_a$ is additive and hence a probability; namely $P(a\cond_f(\cdot)):A\to[0,1]$ is a probability function. We now prove that $(i)$ implies $(ii)$.  Let us reason by contraposition and assume that $f$ doesn't satisfy uniqueness, namely that there is an atom $\alpha$ and an element $a \in A$ such that $|f(a, \alpha)| > 1$ or $f(a, \alpha)=\bot$. If the former holds, let $b = f(a, \alpha)$, then, by definition of $Bel_a$, we would have that $Bel_a(b)>0$, since $m_a(b)>0$.  Thus, take $\beta\in b$ so that $b=\beta\vee (b\wedge\neg \beta)$, and then one has $Bel_a(b) > Bel_a(\beta) + Bel_a(b\wedge \neg \beta)$, whence $Bel_a$ is not additive. 
Finally, if $f(a, \alpha)=\bot$, then $Bel_a(\bot)>0$. Therefore, $Bel_a$ is not a probability either. So, by Fact \ref{fc:characterization}, $P(a\cond_f(\cdot))$ is not a probability.
\end{proof}

The above results are general findings concerning conditionals and probability, and they can be interpreted as having limitative implications. On one hand, they demonstrate that the probability of a selection function-based conditional aligns with imaged belief functions. On the other hand, they establish that the belief function induced by a conditional is a probability function if and only if the selection function underlying the conditional satisfies the uniqueness property. Given these general results, we can now proceed to investigate whether more specific relationships exist between updating methods for probability and selection function-based conditionals.



\section{Dealing with General Updates}\label{sec4}

 As briefly mentioned in the introduction, Lewis \cite{Lewis1976} introduced \emph{imaging} as an alternative to conditionalization for updating a probability measure $P:{\bf A}\to [0,1]$. Roughly, this procedure changes the initial probability $P$ by acting as if a certain event $a$ holds. This update is accomplished by leveraging a similarity structure among possible worlds (or, equivalently, a selection function): the updated $a$-imaged probability $P_a$ is obtained by shifting the probability mass of each atom $\alpha$ to its closest $a$-world. Crucially, this updating procedure relies on a uniqueness assumption: for each  $\alpha\in \atomA$, there is only one closest $a$-world to which $\alpha$ transfers its probability mass $P(\alpha)$. G{\"a}rdenfors \cite{Gardenfors1982} further generalized Lewis's imaging by relaxing this uniqueness assumption, allowing for multiple closest worlds. Building on this prior work and drawing inspiration from G\"uenther \cite{Guenther2022PCIP}, we now introduce a general framework for defining updating methods for probabilities that generalizes both Lewis and G{\"a}rdenfors's imaging. 

First, given a finite Boolean algebra $\mathbf{A}$ and a selection function $f: A \times \atom(\mathbf{A})\to A$, we define a {\em distribution function for $f$} as a map $\lambda:A\times \atom({\bf A})\to [0,1]^{|\atomA|}$ satisfying the following constraints:

\begin{itemize}
\item $\lambda(a, \alpha)(\beta) =0$ if and only if $\beta \notin f(a, \alpha)$

\item $\sum\limits_{\beta \in f(a, \alpha)}\lambda(a, \alpha)(\beta)=1$

\end{itemize}
Observe that the second item above states that, for all $a\in A$ and $\alpha\in \atomA$, $\lambda(a,\alpha)$ is a probability distribution on $f(a,\alpha)$.


Now, consider a finite Boolean algebra $\mathbf{A}$, a positive probability $P:{\bf A}\to[0,1]$, a selection function $f: A \times \atom(\mathbf{A})\to  A$, and distribution function $\lambda:A\times \atom({\bf A})\to [0,1]^{|\atomA|}$ for $f$. Given an element $a \in A\setminus\{\bot\}$, we can define a probability distribution $P^\lambda_a$ over $\atomA$ as follows: for all $\beta \in \atomA$

 \begin{equation}\label{def:im}
P_a^\lambda(\beta)=\sum_{\alpha: \beta\in f(a,\alpha)}\lambda(a,\alpha)(\beta)\cdot P(\alpha)
\end{equation}

Specifically, the probability distribution $P_a^\lambda(\cdot)$ is obtained by redistributing the original probability mass of each atom $\alpha$ among all atoms in $f(a,\alpha)$, resulting in a new probability distribution. Hence, the corresponding updated probability function can be expressed as:

\begin{equation}\label{def:improb}
P_a^\lambda(b)=\sum_{\beta \in b} P^\lambda_a(\beta)=\sum_{\beta \in b}\left(\sum_{\alpha: \beta\in f(a,\alpha)}\lambda(a,\alpha)(\beta)\cdot P(\alpha)\right).
\end{equation}

$P_a^\lambda(\cdot)$ represents an updated probability distribution where the shift of probability mass is performed according to the selection function $f$. The constraints imposed on the distribution function $\lambda$ ensure that $P_a^\lambda(\cdot)$ is indeed a probability distribution. 
Moreover, all the probability mass is transferred to the atoms selected by $f$. Indeed, the following result can be proven from equation \ref{def:improb} and the definition of the distribution function $\lambda$:

\begin{fact}\label{fac:imisprob}

Consider a finite Boolean algebra $\mathbf{A}$, a selection function $f: A \times \atom(\mathbf{A})\to  A$ and a distribution mass $\lambda:A\times \atom({\bf A})\to [0,1]^{|\atomA|}$. For $a \in A\setminus\{\bot\}$, $P^\lambda_a(\cdot)$ is a probability function

\end{fact}


According to the general updating method for probabilities defined above, which utilizes the distribution function $\lambda$, several different updating procedures can be represented, including conditionalization \cite{Guenther2022PCIP}, imaging \cite{Lewis1976}, and general imaging \cite{Gardenfors1982}. Furthermore, since $\lambda$ allows for considerable flexibility in the redistribution of probability mass, numerous other updating procedures, beyond these standard and common ones, can also be defined within this framework.

Having established this foundation, an important question remains: what is the relationship between selection function-based conditionals and $\lambda$-updating methods for probability? In light of the preceding results, we can demonstrate that these updating methods can be characterized as a probability induced by a conditional only in specific cases, thus proving that selection function-based conditionals cannot capture the entire scope of $\lambda$-updating procedures.

First, let us prove a general result concerning the relationship between the probability of conditionals and updating methods, namely that the probability of a conditional $a \cond_f b$ (with $f$ satisfying normality) will always be less than or equal than the corresponding $\lambda$-updated probability of $b$ given $a$:

 \begin{fact}\label{dc:lessorequal}
Consider a finite Boolean algebra $\mathbf{A}$, a selection function $f: A \times \atom(\mathbf{A})\to  A$ satisfying normality, and a distribution mass $\lambda:A\times \atom({\bf A})\to [0,1]^{|\atomA|}$. For $a \in A\setminus\{\bot\}$, and for all $b \in A$, \[P(a \cond_f b) \leq P^\lambda_a(b)\]
 \end{fact}

\begin{proof} (Sketch)
We show that for any $\alpha \in   a \cond_f b$, the initial probability mass of $P(\alpha$) contributes to the calculation of $P^\lambda_a(b)$. Specifically, let us consider $\alpha \in \atomA$ such that $\alpha \in   a \cond_f b$. This implies that $  f(a, \alpha) \subseteq  b$. According to the definition of $P^\lambda_a(\cdot)$, the probability mass $P(\alpha)$ is redistributed among the elements in $f(a, \alpha)$. However, since $  f(a, \alpha) \subseteq  b$, we know that for every $\beta \in f(a, \alpha)$, $\beta \in   b$. This implies that the mass $P(\alpha)$ from every $\alpha \in   a\cond_f b$ is redistributed among atoms that are below $b$. Consequently, this mass will be included in the summation that yields $P^\lambda_a(b)$. Thus, we have $\sum_{\alpha \in   a \cond_f b}P(\alpha) \leq P^\lambda_a(b)$, which is equivalent to $P(a \cond_f b) \leq P^\lambda_a(b)$.
\end{proof}

We now have all the necessary components to show the conditions under which the probability of a conditional corresponds to an updated probability.

\begin{theorem}\label{thmMain1}
Consider a finite Boolean algebra $\mathbf{A}$, a selection function $f:A\times \atom({\bf A})\to A$ satisfying normality, a distribution function $\lambda:A\times \atom({\bf A})\to [0,1]^{|\atomA|}$, and a positive probability $P: A \to [0, 1]$. The following holds: \[(\text{for all $ a \in A\setminus\{\bot\}$ and $b \in A$, } P(a \cond_f b) = P^\lambda_a(b)) \Leftrightarrow \text{ $f$ satisfies uniqueness }\]
\end{theorem}

\begin{proof} (Sketch)
$(\Leftarrow)$. This direction is straightforward. $(\Rightarrow)$ By contraposition, assume that $f$ doesn't satisfy uniqueness, namely there are $\alpha$ and $a$ such that $|f(a, \alpha)| > 1$, by normality. Then, by Proposition \ref{propBasic1}, we have that $P(a \cond_f(\cdot))$ is not a probability function, hence, a fortiori, it cannot be $P^\lambda_a(\cdot)$ which is a probability by Fact \ref{fac:imisprob}; therefore there must be some $a\in A\setminus\{\bot\}$ and $b \in A$ such that $P(a \cond_f b)\neq P^\lambda_a(b)$. 
\end{proof}

In conclusion of this section, let us examine the previous result in more detail. First of all, let us observe that if $f$ satisfies the uniqueness constraint (along with normality), 
then for all $a \in A\setminus\{\bot\}$, $P(a \cond_f (\cdot))$ and $P^\lambda_a(\cdot)$ coincide. This is because, under uniqueness, $\lambda(a,\alpha)$, being a probability distribution on $f(a,\alpha)=\{\beta\}$, it can only take values of 0 or 1; specifically,  $\lambda(a,\alpha)(\beta)=1$ and 0 otherwise. This means that each world (or atom) $\alpha$ transfers its entire probability mass to its unique closest $a$-world, as 
determined by $f$. Hence, due to uniqueness, the mass $m_a$, reduces 
to a probability distribution. 

Conversely, when $f$ does not satisfy uniqueness, it generally holds that for some  $a,b\in A$, with $a\neq\bot$, $P(a \cond_f b) < P^\lambda_a(b)$ as the following example shows.

\begin{example}
Let us consider the algebra ${\bf A}$ of three atoms, i.e. $\atomA=\{\alpha_1,\alpha_2,\alpha_3\}$. Furthermore, conforming to our convention,  let us represent the elements of $A$ as subsets of $\atomA$.

Let us consider the selection function $f$ such that $f(\{\alpha_2, \alpha_3\}, \alpha_1)=\{\alpha_2, \alpha_3\}$, 
 $f(\{\alpha_2, \alpha_3\}, \alpha_2)=\{\alpha_2\}$ and $f(\{\alpha_2, \alpha_3\}, \alpha_3)=\{\alpha_3\}$. Notice that $f$ does satisfy normality, yet not the uniqueness property because $|f(\{\alpha_2, \alpha_3\}, \alpha_1)|=2$. 

 Let $P$ be a positive probability on ${\bf A}$ and $\lambda$ be a distribution function. Recall from what we stated at the very beginning of this section that for all $a\in A$ and $\alpha\in \atom{\bf A}$, $\lambda(a,\alpha)$ is a probability distribution  on $f(a, \alpha)$. Thus, assuming $\lambda$ to be positive, we have that
 \begin{itemize}
 \item $\lambda(\{\alpha_2,\alpha_3\}, \alpha_1)$ distributes postively on $f(\{\alpha_2,\alpha_3\}, \alpha_1)=\{\alpha_2,\alpha_3\}$ and therefore\\ $\lambda(\{\alpha_2,\alpha_3\}, \alpha_1)(\alpha_2)>0$ and $\lambda(\{\alpha_2,\alpha_3\}, \alpha_1)(\alpha_3)>0$;
 \item $\lambda(\{\alpha_2,\alpha_3\}, \alpha_2)$ can only assign $1$ to $\alpha_2$ being ${\alpha_2}$ the unique element in its domain.
  \end{itemize}
  Then, by definition:
 \begin{itemize}
\item $P(\{\alpha_2,\alpha_3\}\cond_f {\alpha_2})=P(f(\{\alpha_2, \alpha_3\}, \alpha_2))=P(\alpha_2)$;
\item $P^\lambda_{\{\alpha_2,\alpha_3\}}(\alpha_2)=\lambda(\{\alpha_2,\alpha_3\}, \alpha_1)(\alpha_2)\cdot P(\alpha_1)+\lambda(\{\alpha_2,\alpha_3\}, \alpha_2)(\alpha_2)\cdot P(\alpha_2)=\lambda(\{\alpha_2,\alpha_3\}, \alpha_1)(\alpha_2)\cdot P(\alpha_1)+1\cdot P(\alpha_2)$.
 \end{itemize}
Therefore, since $P$ is positive and $\lambda(\{\alpha_2,\alpha_3\}, \alpha_1)(\alpha_2)>0$,
$$
P(\{\alpha_2,\alpha_3\}\cond_f {\alpha_2})<P^\lambda_{\{\alpha_2,\alpha_3\}}(\alpha_2).
$$
\end{example}


\section{Conclusion}\label{sec5}

The present work aimed to demonstrate that the probability of selection function conditionals can be characterized in terms of updated belief functions within the framework of Dempster-Shafer theory. Conversely, classical Bayesian probabilistic updates (defined in terms of $\lambda$) cannot be characterized as the probability of a selection function conditional, except in a highly restricted case, specifically when the set of closest worlds involved in defining both the update and the conditional is a singleton. These results are quite general, as selection function conditionals represent a broad class of conditionals, including variably strict conditionals and preferential conditionals.

The results presented here can be further generalized to gain new insights into the logic of conditionals. In particular, the finding that the probability of selection function-based conditionals can be represented as a belief function suggests that this type of conditional might be representable in terms of a standard normal modal operator from modal logic. Indeed, our results in Section 3 can be further generalized to explore whether selection function-based conditionals can be represented as $a \cond_f b \equiv \square (a \rightarrow b)$, where $\to$ can be another binary connective. Previous work in this direction by \cite{Rosella2023-ROSCAM-10} has shown that Lewis counterfactuals $a \cond_f b$ can be represented as $\square(b \mid a)$ within the framework of Boolean algebras of conditionals \cite{Flaminio2020}, where $(\cdot\mid\cdot)$ is a conditional that satisfies suitable algebraic counterparts of the laws of conditional probability. Given that Lewis counterfactuals are a specific type of selection function conditional, our findings in this work suggest that the representation result in \cite{Rosella2023-ROSCAM-10} can also be extended to other kinds of selection function conditionals by imposing appropriate constraints on $\square$.

On a more conceptual level, we have shown that a broad class of updating methods, including conditionalization, cannot be interpreted as the probability that a certain conditional connective holds. A question remains open regarding whether there exists a specific operator whose probability can represent these kinds of updating methods. Our results here indicate that the possible candidates are limited.

\subsection*{Acknowledgments} The authors are grateful to the anonymous referees for their careful reading and valuable suggestions to clarify several parts of the present paper.  
Flaminio and Godo acknowledge the Spanish projects SHORE (PID2022-141529NB-C21)  and LINEXSYS (PID2022-139835NB-C21) respectively, both funded by MCIU/AEI/10.13039/501100011033. 
 Flaminio and Godo also acknowledge partial support by the  H2020-MSCA-RISE-2020 project MOSAIC (Grant Agreement number 101007627).
Giuliano Rosella acknowledges financial support from the Italian Ministry of University and Research (MUR) through the PRIN 2022 grant n. 2022ARRY9N 
funded by the European Union (Next Generation EU).

\nocite{*}
\bibliographystyle{eptcs}

\end{document}